\documentclass[sigconf, natbib=false]{acmart}
\acmConference{FormaliSE}{May 2017}{Buenos Aires, Argentina}
\acmYear{2017}

\usepackage{microtype}
\usepackage{xspace,amsmath,amsfonts,amssymb,bbm,ifthen,mathtools}
\usepackage{bbold}
\usepackage[heavycircles]{stmaryrd}
\usepackage[sort]{cite}
\usepackage[noend]{algpseudocode}
\makeatletter
\renewcommand\ALG@beginblock[1]
   {%
      \ALG@thistlm\ALG@tlm%
      \addtolength\ALG@tlm{#1}%
      \addtolength\ALG@tlm{-1.2ex}%
      \addtocounter{ALG@nested}{1}%
      \setlength\ALG@tmplength{#1}%
      \expandafter\edef\csname ALG@ind@\theALG@nested\endcsname{\the\ALG@tmplength}%
   }%
\renewcommand\ALG@endblock%
   {%
      \addtolength\ALG@tlm{-\csname ALG@ind@\theALG@nested\endcsname}%
      \addtolength\ALG@tlm{1.2ex}%
      \addtocounter{ALG@nested}{-1}%
      \ALG@thistlm\ALG@tlm%
   }%
\makeatother

\theoremstyle{definition}\newtheorem{problem}{Problem}
\newcommand*\ie{\textit{i.e.,}\xspace}
\newcommand*\mcal[1]{\mathcal{#1}}
\newcommand*\Real{\mathbbm{R}}
\newcommand*\Realnn{\Real_{ \ge 0}}
\newcommand*\px{\textit{px}}
\newcommand*\Bool{\mathbbm{B}}
\newcommand*\sem[2][]{\llbracket #2\rrbracket^{#1}}
\newcommand*\proj[2]{\textup{\textsf{proj}}_{#1}( #2)}
\newcommand*\GPart{\textit{GP}}
\newcommand*\GP[1]{\GPart[ #1]}
\newcommand*\limpl{\Rightarrow}
\newcommand*\dom{\textup{dom}}
\newcommand*\ltrue{\textup{\textbf{tt}}}
\newcommand*\lfalse{\textup{\textbf{ff}}}
\newcommand*\Trop{\mathbbm{T}}
\newcommand*\Fuzz{\mathbbm{F}}
\newcommand*\nul{\mathbb{0}}
\newcommand*\one{\mathbb{1}}
\newcommand*\Nrg{\mathbbm{E}}
\newcommand*\E{\mathcal E}
\newcommand*\id{\textup{\textsf{id}}}
\newcommand*\Nat{\mathbbm{N}}
\newcommand*\V{\mathcal V}
\newcommand*\Vrg{\mathbbm{V}}
\newcommand*\cf{\textit{cf.}\xspace}
\newcommand*\bigland{\bigwedge}

\begin{document}

\title{Featured Weighted Automata}

\author{Uli Fahrenberg}%
\authornote{Most of this work was carried out while this author was
  still employed at Inria Rennes.}%
\affiliation{%
  \institution{Ecole polytechnique, Palaiseau, France}}%
\email{uli@lix.polytechnique.fr}%

\author{Axel Legay}
\affiliation{%
  \institution{Inria Rennes, France}}%
\email{axel.legay@inria.fr}

\begin{abstract}
  A featured transition system is a transition system in which the
  transitions are annotated with feature expressions: Boolean
  expressions on a finite number of given features.  Depending
  on its feature expression, each individual transition can be
  enabled when some features are present, and disabled for other sets
  of features.  The behavior of a featured transition system hence
  depends on a given set of features.  There are algorithms for
  featured transition systems which can check their properties for all
  sets of features at once, for example for LTL or CTL properties.

  Here we introduce a model of featured weighted automata which
  combines featured transition systems and (semiring-) weighted
  automata.  We show that methods and techniques from weighted
  automata extend to featured weighted automata and devise algorithms
  to compute quantitative properties of featured weighted automata for
  all sets of features at once.  We show applications to minimum
  reachability and to energy properties.
\end{abstract}

\maketitle

\section{Introduction}

A \emph{featured transition
  system}~\cite{DBLP:journals/tse/ClassenCSHLR13} is a transition
system in which the transitions are annotated with \emph{feature
  expressions}: Boolean expressions involving a finite number of given
features.  Depending on its feature expression, each individual
transition can be enabled when some features are present, and disabled
for other sets of features.  For any set of features, a given featured
transition system \emph{projects} to a transition system which
contains precisely the transitions which are enabled for that set of
features.

Standard problems such as reachability or safety can be posed for
featured transition systems, where the interest now is to check these
properties \emph{for all sets of features at once}.  Hence, for
example for reachability, given a featured transition system and a set
of accepting states, one wants to construct a feature expression
$\phi$ such that an accepting state is reachable iff the set of
features satisfies $\phi$.

For \emph{quantitative} properties of transition systems, the model of
\emph{(semiring-) weighted automata} has proven
useful~\cite{book/DrosteKV09}.  This provides a uniform framework to
treat problems such as minimum reachability, maximum flow, energy
problems~\cite{DBLP:journals/corr/EsikFL15a}, and others.  Here we
extend techniques from weighted automata to \emph{featured weighted
  automata}, \ie~weighted automata in which the transitions are
annotated with feature expressions.  This extension makes it possible
to check quantitative properties for all sets of features at once.

To be precise, a featured transition system induces a (projection)
function from sets of features to transition systems, mapping each set
of features to the behavior under these features.  Similarly, we will
define projections of featured weighted automata, mapping sets of
features to weighted automata.  \emph{Values} of weighted automata are
an abstract encoding of their behavior; we will see how to compute
values of featured weighted automata as functions from feature
expressions to behaviors.

We also develop an application of our techniques to featured
\emph{energy problems}.  Energy problems are important in areas such
as embedded systems or autonomous systems.  They are concerned with
the question whether a given system admits infinite schedules during
which (1) certain tasks can be repeatedly accomplished and (2) the
system never runs out of energy (or other specified resources).
Starting with~\cite{DBLP:conf/formats/BouyerFLMS08}, formal modeling
and analysis of such problems has attracted some
attention~\cite{conf/ictac/FahrenbergJLS11, DBLP:conf/lata/Quaas11,
  DBLP:conf/icalp/ChatterjeeD10, DBLP:conf/hybrid/BouyerFLM10,
  DBLP:conf/atva/EsikFLQ13, DBLP:conf/qest/BouyerLM12,
  DBLP:conf/csl/DegorreDGRT10}. 

Featured transition systems have applications in \emph{software
  product lines}, where they are used as abstract representations of
the behaviors of variability
models~\cite{DBLP:journals/csur/ThumAKSS14}.  This representation
allows one to analyze all behaviors of a software product line at
once, as opposed to analyzing each product on its own.  Similarly,
featured weighted automata can be used as abstract representations of
quantitative behaviors of software product lines, and the present work
enables analysis of quantitative behaviors of all products in a
software product line at once.

\vspace*{-1ex}
\paragraph{Contributions and structure of the paper.}

We start in Sect.~\ref{se:minreach} by revisiting minimum reachability
in featured transition systems with transitions weighted by real
numbers.  This has to some extent already been done
in~\cite{DBLP:conf/icse/CordySHL13}, but we reformulate it in order to
prepare for the generalization in the following sections.

In Sect.~\ref{se:fwa}, we introduce featured weighted automata and
show some first examples.  Instead of semirings, we will work with
(featured) automata weighted in \emph{$^*$-continuous Kleene
  algebras}; this is for convenience of presentation only, and all our
work (except for Sect.~\ref{se:energy}) can be extended to a more
general (for example non-idempotent) setting.
In Sect.~\ref{se:fba-fwa}, we then show how methods and techniques
from weighted automata can be transferred to featured weighted
automata.

In the last Sect.~\ref{se:energy}, we extend our results to develop an
application to featured energy problems.  This is based on the recent
result in~\cite{DBLP:journals/corr/EsikFL15a} that energy problems can
be stated as B{\"u}chi problems in automata weighted in
\emph{$^*$-continuous Kleene $\omega$-algebras}, which are certain
types of semimodules over $^*$-continuous Kleene algebras; hence we
need to extend our results to such semimodules.

The paper is followed by a separate appendix which contains some of
the proofs of our results.

\section{Minimum Reachability in Real-Weighted Featured Automata}
\label{se:minreach}


A \emph{real-weighted automaton} $\mcal S=( S, I, F, T)$ consists of a
finite set $S$ of states, subsets $I, F\subseteq S$ of initial and
accepting states, and a finite set $T\subseteq S\times\Realnn\times S$
of weighted transitions.  Here $\Realnn$ denotes the set of
non-negative real numbers.

A \emph{finite path} in such a real-weighted automaton $\mcal S$ is a
finite alternating sequence
$\pi=( s_0, x_0, s_1, x_1,\dotsc, x_k, s_{ k+ 1})$ of transitions
$( s_0, x_0, s_1),\dotsc,( s_k, x_k, s_{ k+ 1})\in T$.  The
\emph{weight} of $\pi$ is the sum
$w( \pi)= x_0+\dotsm+ x_k\in \Realnn$.  A finite path $\pi$ as above
is said to be \emph{accepting} if $s_0\in I$ and $s_{ k+ 1}\in F$.
The \emph{minimum reachability problem} for real-weighted automata
asks, given a real-weighted automaton $\mcal S$ as above, to compute
the value
\begin{equation*}
  | \mcal S|= \inf\{ w( \pi)\mid \pi\text{ accepting finite path in $\mcal
    S$}\}\,.
\end{equation*}
That is, $| \mcal S|$ is the minimum weight of all finite paths from
an initial to an accepting state in $\mcal S$.  This being a
multi-source-multi-target shortest path problem, it can for example be
solved using the Floyd-Warshall relaxation algorithm.


Let $N$ be a set of \emph{features} and $\px\subseteq 2^N$ a set of
\emph{products} over $N$.  A \emph{feature guard} is a Boolean
expression over $N$, and we denote the set of these by $\Bool( N)$.
We write $p\models \gamma$ if
$p\in \px$ satisfies
$\gamma\in \Bool( N)$ and
$\sem \gamma=\{ p\in \px\mid p\models \gamma\}$.  Note that
$\sem \gamma$ is a set of sets of features.

\begin{definition}
  A \emph{real-weighted featured automaton} $( S, I, F, T, \gamma)$
  consists of a finite set $S$ of states, subsets $I, F\subseteq S$ of
  initial and accepting states, a finite set
  $T\subseteq S\times\Realnn\times S$ of weighted transitions, and a
  feature guard mapping $\gamma: T\to \Bool( N)$.
\end{definition}

The \emph{projection} of a real-weighted featured automaton $\mcal F=$
\linebreak $( S, I, F, T, \gamma)$ to a product $p\in \px$ is the
real-weighted automaton $\proj p{ \mcal F}=( S, I, F, T')$ with
$T'=\{ t\in T\mid p\models \gamma( t)\}$.

For each product $p\in \px$, we could solve the shortest path problem
in $\proj p{ \mcal F}$ by computing $| \proj p{ \mcal F}|$.  Instead,
we develop an algorithm which computes all these values at the same
time.  Its output will, thus, be a function
$| \mcal F|: \px\to \Realnn$, with the property that for every
$p\in \px$, $| \mcal F|( p)=| \proj p{ \mcal F}|$.

As a symbolic representation of functions $\px\to \Realnn$, we use
injective functions from \emph{guard partitions} to $\Realnn$.
Intuitively, a guard partition is a set of feature guards which
partitions $\px$ into classes such that within each class, $f$ has the
same value for all products, and between different classes, $f$ has
different values.

\begin{definition}
  A \emph{guard partition} of $\px$ is a set $P\subseteq \Bool( N)$
  such that $\sem{ \bigvee P}= \px$, $\sem \gamma\ne \emptyset$ for
  all $\gamma\in P$, and
  $\sem{ \gamma_1}\cap \sem{ \gamma_2}= \emptyset$ for all
  $\gamma_1, \gamma_2\in P$ with $\gamma_1\ne \gamma_2$.  The set of
  all guard partitions of $\px$ is denoted
  $\GPart\subseteq 2^{ \Bool( N)}$.
\end{definition}

A guard partition is a logical analogue to a partition of the set of
products $\px$: any guard partition induces a partition of $\px$, and
any partition of $\px$ can be obtained by a guard partition.  In
particular, for any guard partition $P$ and any product $\px$, there
is precisely one $\gamma\in P$ for which $\px\models \gamma$.

Let
$\GP \Realnn=\{ f: P\to \Realnn\mid P\in \GPart, \forall \gamma_1,
\gamma_2\in P: \gamma_1\ne \gamma_2\limpl f( \gamma_1)\ne f(
\gamma_2)\}$
denote the set of injective functions from guard partitions to
$\Realnn$.

We use \emph{injective} functions $P\to \Realnn$ as symbolic
representations of functions $\px\to \Realnn$, because they provide
the most \emph{concise} such representation.  Indeed, if a function
$f: P\to \Realnn$ is not injective, then there are feature guards
$\gamma_1, \gamma_2\in P$ for which $f( \gamma_1)= f( \gamma_2)$, so
we can obtain a more concise representation of $f$ by letting
$P'= P\setminus\{ \gamma_1, \gamma_2\}\cup\{ \gamma_1\lor \gamma_2\}$
and $f': P'\to \Realnn$ be defined by $f'( \delta)= f( \delta)$ for
$\delta\ne \gamma_1\lor \gamma_2$ and
$f'( \gamma_1\lor \gamma_2)= f( \gamma_1)$.

\begin{figure}[tbp]
  \begin{algorithmic}[1]
    \State \textbf{Input:} real-weighted featured automaton
    $\mcal F=( S, I, F, T, \gamma)$ with $S=\{ s_1,\dotsc, s_n\}$
    \State \textbf{Output:} function $| \mcal F|\in \GP \Realnn$
    \Statex 
    \State \textbf{var} $D:\{ 1,\dotsc, n\}\times\{ 1,\dotsc, n\}\to
    \GP \Realnn$
    \State \textbf{var} $P$, $f$
    \For {$i\gets 1$ to $n$}
    \For {$j\gets 1$ to $n$}
    \State $\dom( D( i, j))\gets\{ \ltrue\}$
    \State $D( i, j)( \ltrue)\gets \infty$
    \ForAll {$( s_i, x, s_j)\in T$}
    \ForAll {$\gamma\in \dom( D( i, j))$}
    \If {$\sem{ \gamma\land \gamma( s_i, x, s_j)}\ne \emptyset$ and
      $D( i, j)( \gamma)> x$}
    \State \textsc{Split}($D( i, j), \gamma, \gamma( s_i, x, s_j), x$)
    \EndIf
    \EndFor
    \EndFor
    \EndFor
    \EndFor
    \For {$i\gets 1$ to $n$} \label{al:minreach.fw.start}
    \For {$j\gets 1$ to $n$}
    \For {$k\gets 1$ to $n$}
    \State \textsc{Relax}($i, j, k$)
    \EndFor
    \EndFor
    \EndFor \label{al:minreach.fw.end}
    \State $P\gets\{ \ltrue\}$; $f( \ltrue)\gets
    \infty$ \label{al:minreach.findmin.start}
    \ForAll {$s_i\in I$}
    \ForAll {$s_j\in F$}
    \ForAll {$\gamma_1\in P$}
    \ForAll {$\gamma_2\in \dom( D( i, j))$}
    \If {$\sem{ \gamma_1\land \gamma_2}\ne \emptyset$ and $f(
      \gamma_1)> D( i, j)( \gamma_2)$}
    \State \textsc{Split}($f, \gamma_1, \gamma_2, D( i, j)( \gamma_2)$)
    \EndIf
    \EndFor
    \EndFor
    \EndFor
    \EndFor
    \State \Return $f$
    \Statex
    \Procedure{Relax}{$i,j,k$} \label{al:minreach.relax.start}
    \ForAll {$\gamma_1\in \dom( D( i, j))$}
    \ForAll {$\gamma_2\in \dom( D( i, k))$}
    \ForAll {$\gamma_3\in \dom( D( k, j))$}
    \If {$\sem{ \gamma_1\land \gamma_2\land \gamma_3}\ne
      \emptyset$}
    \If {$D( i, j)( \gamma_1)> D( i, k)( \gamma_2)+ D( k, j)(
      \gamma_3)$} \label{al:minreach.relax.notempty}
    \State \textsc{Split}($D( i, j), \gamma_1, \gamma_2\land
    \gamma_3, D( i, k)( \gamma_2)+ D( k, j)( \gamma_3)$)
    \EndIf
    \EndIf
    \EndFor
    \EndFor
    \EndFor
    \EndProcedure \label{al:minreach.relax.end}
    \Statex
    \Procedure{Split}{$f: P\to \Realnn, \gamma_1, \gamma_2\in \Bool(
      N), x\in \Realnn$} \label{al:minreach.split.start}
    \If {$\sem{ \gamma_1}= \sem{ \gamma_1\land
        \gamma_2}$} \label{al:minreach.split.cond}
    \State $f( \gamma_1)\gets x$ \label{al:minreach.split.nosplit.assign}
    \State \textsc{Combine}($f, \gamma_1$)
    \Else \label{al:minreach.split.split}
    \State $y\gets f( \gamma_1)$
    \State $P\gets P\setminus\{ \gamma_1\}\cup\{ \gamma_1\land
    \gamma_2, \gamma_1\land \neg \gamma_2\}$
    \State $f( \gamma_1\land \neg \gamma_2)\gets y$
    \State $f( \gamma_1\land \gamma_2)\gets x$
    \State \textsc{Combine}($f, \gamma_1\land \gamma_2$)
    \EndIf
    \EndProcedure \label{al:minreach.split.end}
    \Statex
    \Procedure{Combine}{$f: P\to \Realnn, \gamma\in \Bool( N)$}
    \State $x\gets f( \gamma)$
    \ForAll {$\delta\in P\setminus\{
      \gamma\}$} \label{al:minreach.combine.forall}
    \If {$f( \delta)= f( \gamma)$}
    \State $P\gets P\setminus\{ \delta, \gamma\}\cup\{ \delta\lor \gamma\}$
    \label{al:minreach.combine.join}
    \State $f( \delta\lor \gamma)\gets x$
    \State \textbf{break}
    \EndIf
    \EndFor
    \EndProcedure
  \end{algorithmic}
  \caption{
    \label{fi:alg-minreach}
    Algorithm to compute $| \mcal F|$ for a real-weighted featured
    automaton $\mcal F$.}
\end{figure}

We show in Fig.~\ref{fi:alg-minreach} an algorithm to compute a
symbolic representation of $| \mcal F|$.  The algorithm performs, in
lines~\ref{al:minreach.fw.start} to~\ref{al:minreach.fw.end}, a symbolic
Floyd-Warshall relaxation to compute a matrix $D$ which as entries
$D( i, j)$ has functions in $\GP \Realnn$ that for each product return
the shortest path from state $s_i$ to state $s_j$.

The relaxation procedure \textsc{Relax}($i, j, k$) is performed
by comparing $D( i, j)$ to the sum $D( i, k)+ D( k, j)$ and updating
$D( i, j)$ if the sum is smaller.  The result of the comparison
depends on the products for which the different paths are enabled,
hence the comparison and update are done for each feature expression
$\gamma_1$ in the partition for $D( i, j)$ and all feature expressions
$\gamma_2$, $\gamma_3$ in the partitions for $D( i, k)$ and
$D( k, j)$, respectively.  The comparison has to be done only if these
partitions overlap (line~\ref{al:minreach.relax.notempty}), and in
case the sum is smaller, $D( i, j)$ is updated in a call to a
split-and-combine procedure.

Using the procedure \textsc{Split}, in
lines~\ref{al:minreach.split.start} to~\ref{al:minreach.split.end},
$D( i, j)$ is updated at the $\gamma_1\land( \gamma_2\land \gamma_3)$
part of its partition.  If
$\sem{ \gamma_1\land( \gamma_2\land \gamma_3)}$ is not smaller than
$\sem{ \gamma_1}$ (line~\ref{al:minreach.split.cond}), then
$D( i, j)( \gamma_1)$ is set to its new value.  Afterwards, we need to
call a \textsc{Combine} procedure to see whether $D( i, j)$ has the
same value at any other part $\delta$ of its partition
(line~\ref{al:minreach.combine.forall}) and, in the affirmative case,
to update the partition of $D( i, j)$ by joining the two parts
(line~\ref{al:minreach.combine.join}f).

If the feature expression $\gamma_1\land( \gamma_2\land \gamma_3)$ on
which to update $D( i, j)$ is a strict subset of $\gamma_1$
(line~\ref{al:minreach.split.split}), then the $\gamma_1$ part of the
partition of $D( i, j)$ needs to be split into two parts:
$\gamma_1\land( \gamma_2\land \gamma_3)$, on which $D( i, j)$ is to be
updated, and $\gamma_1\land \neg( \gamma_2\land \gamma_3)$, on which
its value stays the same.  Again, we need to call the \textsc{Combine}
procedure afterwards to potentially combine feature expressions in the
partition of $D( i, j)$.

Once relaxation has finished in line~\ref{al:minreach.findmin.start},
we need to find $f:= \min\{ D( i, j)\mid s_i\in I, s_j\in F\}$.  As
this again depends on which features are present, we need to compute
this minimum in a way similar to what we did in the \textsc{Relax}
procedure: for each feature expression in the partition $P$ of $f$ and
each overlapping feature expression in the partition of $D( i, j)$, we
compare the two values and use the \textsc{Split} procedure to update
$f$ if $D( i, j)$ is smaller.

A variant of the algorithm in Fig.~\ref{fi:alg-minreach} has been
implemented in~\cite{DBLP:conf/splc/OlaecheaFAL16}, as part of an effort to compute
minimum limit-average cost in real-weighted featured automata.
Several experiments in~\cite{DBLP:conf/splc/OlaecheaFAL16} show that our algorithm is
significantly faster than an approach which separately solves the
minimum reachability problem for each product.

\section{Featured Weighted Automata}
\label{se:fwa}

We proceed to introduce a generalization of the setting in the
previous section.  Here $\Realnn$ is replaced by an abstract
\emph{$^*$-continuous Kleene algebra}.  This allows us to develop an
abstract setting for analysis of \emph{featured weighted automata},
and to re-use our techniques developed
in the previous section to solve quantitative problems in other
concrete settings.
%

\subsection{Weighted Automata}

Recall that a \emph{semiring}~\cite{book/DrosteKV09}
$K=( K, \oplus, \otimes, 0, 1)$ consists of a commutative monoid
$( K, \oplus, 0)$ and a monoid $( K, \otimes, 1)$ such that the
distributive and zero laws
\begin{equation*}
  x( y\oplus z)= x y\oplus x z \qquad ( y\oplus z) x= y x\oplus
  z x \qquad 0\otimes x= 0= x\otimes 0
\end{equation*}
hold for all $x, y, z\in K$ (here we have omitted the multiplication
sign $\otimes$ in some expressions, and we shall also do so in the
future).  It follows that the product distributes over all finite
sums.

A (finite) \emph{weighted automaton}~\cite{book/DrosteKV09} over a
semiring $K$ (or a \emph{$K$-weighted automaton} for short) is a tuple
$\mcal S=( S, I, F, T)$ consisting of a finite set $S$ of states, a subset
$I\subseteq S$ of initial states, a subset $F\subseteq S$ of
accepting states, and a finite set $T\subseteq S\times K\times S$ of
transitions.

A \emph{finite path} in such a $K$-weighted automaton
$\mcal S=( S, I, F, T)$ is a finite alternating sequence
$\pi=( s_0, x_0, s_1,\dotsc, x_k, s_{ k+ 1})$ of transitions
$( s_0, x_0, s_1),\dotsc,( s_k, x_k, s_{ k+ 1})\in T$.  The
\emph{weight} of~$\pi$ is the product $w( \pi)= x_0\dotsm x_k\in K$.
A finite path $\pi$ as above is said to be \emph{accepting} if
$s_0\in I$ and $s_{ k+ 1}\in F$.  The \emph{reachability value}
$| \mcal S|$ of $\mcal S$ is defined to be the sum of the weights of
all its accepting finite paths:
\begin{equation*}
  | \mcal S|= \bigoplus\{ w( \pi)\mid \pi\text{ accepting finite path
    in $\mcal S$}\}
\end{equation*}
As the set of accepting finite paths generally will be infinite, one
has to assume that such sums exist in $K$ for this definition to make
sense.  This is the subject of Sect.~\ref{se:scka} below.

\subsection{Examples}
\label{se:examples}

The \emph{Boolean semiring} is
$\Bool=(\{ \lfalse, \ltrue\}, \mathord\lor, \mathord\land, \lfalse,
\ltrue)$,
with disjunction as $\oplus$ and conjunction as $\otimes$.  A
$\Bool$-weighted automaton $\mcal S$ hence has its transitions
annotated with $\lfalse$ or $\ltrue$.  For a finite path
$\pi=( s_0, x_0, s_1,\dotsc, x_k, s_{ k+ 1})$, we have
$w( \pi)= \ltrue$ iff all $x_0=\dotsm= x_k= \ltrue$.  Hence
$| \mcal S|= \ltrue$ iff there exists an accepting finite path in
$\mcal S$ which involves only $\ltrue$-labeled transitions.  That is,
$\Bool$-weighted automata are equivalent to ordinary (unlabeled)
automata, where the equivalence consists in removing all
$\lfalse$-labeled transitions.

The \emph{tropical semiring} is
$\Trop=( \Realnn\cup\{ \infty\}, \mathord\wedge, \mathord+, \infty,
0)$, where $\Realnn\cup\{ \infty\}$
denotes the set of extended real numbers, with minimum as $\oplus$ and
addition as $\otimes$.  The weight of a finite path is now the sum of
its transition weights, and the reachability value of a
$\Trop$-weighted automaton is the minimum of all its accepting finite
paths' weights.  Hence $\Trop$-weighted automata are precisely the
real-weighted automata of Sect.~\ref{se:minreach}, and to compute
their reachability values is to solve the \emph{minimum reachability}
problem.

The \emph{fuzzy semiring} is
$\Fuzz=( \Realnn\cup\{ \infty\}, \mathord\vee, \mathord\wedge, 0,
\infty)$,
with maximum as $\oplus$ and minimum as $\otimes$.  Here, the weight
of a finite path is the minimum of its transition weights, and the
reachability value of an $\Fuzz$-weighted automaton is the maximum of
all its accepting finite paths' weights.  This value is hence the
\emph{maximum flow} in a weighted automaton: the maximum available
capacity along any finite path from an initial to a accepting state.

\subsection{Featured Weighted Automata}

We now extend weighted automata with features, for modeling
quantitative behavior of software product lines.
%
Let $K$ be a semiring and denote by
$\GP K=\{ f: P\to K\mid P\in \GPart, \forall \gamma_1, \gamma_2\in P:
\gamma_1\ne \gamma_2\limpl f( \gamma_1)\ne f( \gamma_2)\}$
the set of injective functions from guard partitions to $K$.

\begin{definition}
  A \emph{featured weighted automaton} over $K$ and $\px$ is a tuple
  $( S, I, F, T)$ consisting of a finite set $S$ of states, subsets
  $I, F\subseteq S$ of initial and accepting states, and a finite set
  $T\subseteq S\times \GP K\times S$ of transitions.
\end{definition}

Similarly to what we did in Sect.~\ref{se:minreach}, the transition
labels in $\GP K$ are to be seen as syntactic representations of
functions from products to $K$; we will say more about this below.


\begin{example}
  For $K= \Bool$ the Boolean semiring, featured $\Bool$-weighted
  automata are standard (unlabeled) featured automata: for any feature
  guard $\gamma\in \Bool( N)$, $\{ \gamma, \neg \gamma\}$ is a guard
  partition of $\px$, moreover, for $K= \Bool$, any mapping in $\GP K$
  is equivalent to one from such a guard partition.  Hence transitions
  labeled with feature guards (as in standard featured automata) are
  the same as transitions labeled with functions from guard partitions
  to $\{ \lfalse, \ltrue\}$.
\end{example}

\begin{definition}
  For $f: P\to K\in \GP K$ and $p\in \px$, let $\gamma\in P$ be the
  unique feature guard for which $p\models \gamma$ and define
  $\sem f( p)= f( \gamma)$.  This defines the \emph{semantic
    representation of $f$} as the function $\sem f: \px\to K$.
\end{definition}

\begin{definition}
  Let $\mcal F=( S, I, F, T)$ be a featured $K$-weighted automaton and
  $p\in \px$.  The \emph{projection} of $\mcal F$ to $p$ is the
  $K$-weighted automaton $\proj p{ \mcal F}=( S, I, F, T')$, where
  $T'=\{( s, \sem f( p), s')\mid( s, f, s')\in T\}$.
\end{definition}

The behavior of a featured $K$-weighted automaton is hence given
relative to products: given a featured $K$-weighted automaton
$\mcal F$ and a product $p$, $| \proj p{ \mcal F}|$, provided that it
exists, will be the behavior of $\mcal F$ when restricted to the
particular product $p$.  The purpose of this paper is to show how the
values $| \proj p{ \mcal F}|$ can be computed for all $p\in \px$ at
once.

\subsection{$^*$-Continuous Kleene Algebras}
\label{se:scka}

We finish this section by introducing extra structure and properties
into our semiring $K$ which will ensure that the infinite sums
$| \mcal S|$ always exist.  This is for convenience only, and all our
work can be extended to a more general (for example non-idempotent)
setting.
Recall that a semiring $K=( K, \oplus, \otimes, 0, 1)$ is
\emph{idempotent}~\cite{book/DrosteKV09} if $x\oplus x= x$ for every
$x\in K$.

A \emph{$^*$-continuous Kleene
  algebra}~\cite{DBLP:journals/iandc/Kozen94} is an idempotent
semiring $K=( K, \oplus, \otimes, 0, 1)$ in which all infinite sums of
the form $\bigoplus_{ n\ge 0} x^n$, $x\in K$, exist, and such that
\begin{equation}
  \label{eq:loopabs}
  x\big( \bigoplus_{ n\ge 0} y^n\big) z= \bigoplus_{ n\ge 0} x y^n z
\end{equation}
for all $x, y, z\in K$.  Intuitively, automata weighted over a
$^*$-continuous Kleene algebra allow for \emph{loop abstraction}, in
that the global effects of a loop (right-hand side
of~\eqref{eq:loopabs}) can be computed locally (left-hand side
of~\eqref{eq:loopabs}).  In any $^*$-continuous Kleene algebra $K$ one
can define a unary \emph{star} operation $\mathord{^*}: K\to K$ by
$x^*= \bigoplus_{ n\ge 0} x^n$.

For any semiring $K$ and $n\ge 1$, we can form the \emph{matrix
  semiring} $K^{ n\times n}$ whose elements are $n$-by-$n$ matrices of
elements of $K$ and whose sum and product are given as the usual
matrix sum and product.  It is known~\cite{DBLP:conf/mfcs/Kozen90}
that when $K$ is a $^*$-continuous Kleene algebra, then
$K^{ n\times n}$ is also a $^*$-continuous Kleene algebra, with the
$^*$-operation defined by
$M^*_{ i, j}= \smash{\bigoplus_{ m\ge 0}} \bigoplus\{ M_{ k_1,
  k_2}\dotsm M_{ k_{ m- 1}, k_m}\mid 1\le k_1,\dotsc, k_m\le n,
k_1= i, k_m= j\}$
for all $M\in K^{ n\times n}$ and $1\le i, j\le n$.  Also, if $n\ge 2$
and
$M= \left[ \begin{smallmatrix} a & b \\ c & d \end{smallmatrix}
\right]$,
where $a$ and $d$ are square matrices of dimension less than $n$, then
\begin{equation}
  \label{eq:mstar}
  M^*= 
  \begin{bmatrix}
    ( a\oplus b d^* c)^* & ( a\oplus b d^* c)^* b d^* \\
    ( d\oplus c a^* b)^* c a^* & ( d\oplus c a^* b)^*
  \end{bmatrix}
  .
\end{equation}

The \emph{matrix representation}~\cite{book/DrosteKV09} of a
$K$-weighted automaton $\mcal S=( S, I, F, T)$, with $n= \# S$ the
number of states, is given by the triple $( \alpha, M, k)$, where
$\alpha\in\{ 0, 1\}^n$ is the \emph{initial vector},
$M\in K^{ n\times n}$ is the \emph{transition matrix}, and
$0\le k\le n$.  These are given as follows: order $S=\{ 1,\dotsc, n\}$
such that $i\in F$ iff $i\le k$, \ie~such that the first $k$ states
are accepting, and define $\alpha$ and $M$ by $\alpha_i= 1$ iff
$i\in I$ and $M_{ i, j}= \bigoplus\{ x\mid( i, x, j)\in T\}$.

It can be shown~\cite{inbook/EsikK09} that if $\mcal S$ is a weighted
automaton over a $^*$-continuous Kleene algebra, then the reachability
value of $\mcal S$ is defined and $| \mcal S|= \alpha M^* \kappa$,
where $\kappa\in\{ 0, 1\}^n$ is the vector given by $\kappa_i= 1$ for
$i\le k$ and $\kappa_i= 0$ for $i> k$.

\begin{example}
  Our example semirings $\Bool$, $\Trop$ and $\Fuzz$ share the
  property of being \emph{bounded}.  In general terms, a semiring $K$
  is said to be bounded~\cite{book/DrosteKV09} if $x\oplus 1= 1$ for
  all $x\in K$.  Note that this implies idempotency: for all $x\in K$,
  $x\oplus x= x( 1\oplus 1)= x\otimes 1= x$.  If $K$ is bounded, then
  $x^*= 1\oplus\dotsm= 1$ for all $x\in K$, and $K$ is a
  $^*$-continuous Kleene algebra~\cite{journals/mfm/EsikK02}.

  In $\Bool$, $x\oplus 1= x\lor \ltrue= \ltrue$; in $\Trop$,
  $x\oplus 1= x\land 0= 0$; and in $\Fuzz$,
  $x\oplus 1= x\lor \infty= \infty$; so these three semirings are
  indeed bounded.  Operationally, the fact that $x^*= 1$ means that
  \emph{loops can be disregarded}: for all $x, y, z\in K$,
  $\bigoplus_{ n\ge 0} x y^n z= x y^* z= x z$.  In lieu of the
  examples in Sect.~\ref{se:examples}, it is clear that this property
  holds for $\Bool$-, $\Trop$- and $\Fuzz$-weighted automata: for
  reachability, loops are unimportant; for minimum reachability,
  likewise; and for maximum flow, taking a loop can only decrease the
  flow, hence would be disadvantageous.
\end{example}

\section{Analysis of Featured Weighted Automata}
\label{se:fba-fwa}

Let $K$ be a $^*$-continuous Kleene algebra.  In this section we take
a closer look at the functions in $\GP K$ and define semiring
operations on them.  We show that with these operations, $\GP K$
itself is a $^*$-continuous Kleene algebra.  This means that we can
treat featured $K$-weighted automata as $\GP K$-weighted automata.

We first need to define an operation on partitions which turns
functions $f: P\to K$ from a partition $P\in \GPart$ into
\emph{injective} functions, providing the most concise representation,
by changing their domain.  Intuitively, this \emph{canonicalization}
of $f$ changes the partition $P$ into a coarser one by forming
disjunctions of feature guards on which $f$ has the same value:

\begin{definition}
  \label{de:canon}
  Let $P\in \GPart$ and $f: P\to K$.  Introduce an equivalence
  relation $\mathord\sim\subseteq P\times P$ by
  $\gamma_1\sim \gamma_2$ iff $f( \gamma_1)= f( \gamma_2)$ and let
  $P'= P/ \mathord\sim$ be the quotient.  Let
  $\tilde P=\{ \bigvee \Gamma\mid \Gamma\in P'\}$, then
  $\tilde P\in \GPart$.  For every $\tilde \gamma\in \tilde P$ there
  is an equivalence class $\Gamma\in P'$ for which
  $\tilde \gamma= \bigvee \Gamma$, and $f$ passes to these equivalence
  classes by definition, so we can define $\tilde f: \tilde P\to K$,
  the \emph{canonicalization} of $f$, by
  $\tilde f( \tilde \gamma)= f( \Gamma)$.
\end{definition}

\begin{figure}[tbp]
  \begin{algorithmic}[1]
    \Function{$K$Combine}{$f: P\to K$}: $\GP K$
    \State \textbf{var} $\tilde f$, $\tilde P$
    \State $\tilde P\gets \emptyset$
    \While {$P\ne \emptyset$}
    \State Pick and remove $\gamma$ from $P$
    \State $x\gets f( \gamma)$
    \ForAll {$\delta\in P$}
    \If {$f( \delta)= x$}
    \State $\gamma\gets \gamma\lor \delta$
    \State $P\gets P\setminus\{ \delta\}$
    \EndIf
    \EndFor
    \State $\tilde P\gets \tilde P\cup\{ \gamma\}$
    \State $\tilde f( \gamma)\gets x$
    \EndWhile
    \State \Return $\tilde f: \tilde P\to K$
    \EndFunction
  \end{algorithmic}
  \caption{
    \label{fi:alg-combine}
    Function which computes canonicalization.}
\end{figure}

We show an algorithm which implements canonicalization in
Fig.~\ref{fi:alg-combine}.  The function $K$\textsc{Combine} takes as
input a function $f: P\to K$ and builds its canonicalization
$\tilde f: \tilde P\to K$ by taking disjunctions of feature
expressions in the partition $P$.  Note the similarity of its inner
loop to the \textsc{Combine} procedure of Fig.~\ref{fi:alg-minreach}:
the procedure in Fig.~\ref{fi:alg-minreach} only updates the partition
of $f$ in one place, whereas $K$\textsc{Combine} needs to check the
whole partition.

\begin{lemma}
  \label{le:canon}
  Let $P\in \GPart$, $f: P\to K$, and $\tilde f: \tilde P\to K$ the
  canonicalization of $f$.  Then $\tilde f$ is injective, hence
  $\tilde f\in \GP K$.  Also, for any $\gamma\in P$ there is a unique
  element $\tilde \gamma\in \tilde P$ such that $\sem \gamma\subseteq
  \sem{ \tilde \gamma}$.
\end{lemma}

\begin{definition}
  \label{de:gp-land}
  Let $P_1, P_2\in \GPart$.  The \emph{intersection} of $P_1$ and $P_2$
  is the partition $P= P_1\wedge P_2\in \GPart$ given as $P=\{
  \gamma_1\land \gamma_2\mid \gamma_1\in P_1, \gamma_2\in P_2, \sem{
    \gamma_1\land \gamma_2}\ne \emptyset\}$.
\end{definition}

\begin{lemma}
  \label{le:gp-land-unique}
  Let $P_1, P_2\in \GPart$ and $\gamma\in P_1\wedge P_2$.  There are
  unique elements $\gamma_1\in P_1$, $\gamma_2\in P_2$ such that
  $\gamma= \gamma_1\land \gamma_2$.
\end{lemma}

We can hence write the elements of $P_1\wedge P_2$ as
$\gamma_1\land \gamma_2$ without ambiguity.  We are ready to define
operations $\oplus$, $\otimes$ and $^*$ on functions in $\GP K$.

\begin{definition}
  \label{de:op_gpk}
  Let $f_1: P_1\to K, f_2: P_2\to K\in \GP K$.  Define functions
  $s', p': P_1\wedge P_2\to K$ and $t': P_1\to K$ by
  $s'( \gamma_1\land \gamma_2)= f_1( \gamma_1)\oplus f_2( \gamma_2)$,
  $p'( \gamma_1\land \gamma_2)= f_1( \gamma_1)\otimes f_2( \gamma_2)$,
  and $t'( \gamma_1)= f_1( \gamma_1)^*$.  Let $s, p, t\in \GP K$ be
  the canonicalizations of $s'$, $p'$ and $t'$, respectively, then we
  define $f_1\oplus f_2= s$, $f_1\otimes f_2= p$, and $f_1^*= t$.
\end{definition}

\begin{figure}[tb]
  \begin{algorithmic}[1]
    \Function{$K$Sum}{$f_1: P_1\to K, f_2: P_2\to K$}: $\GP K$
    \State \textbf{var} $f'$, $P'$
    \State $P'\gets \emptyset$
    \ForAll {$\gamma_1\in P_1$}
    \ForAll {$\gamma_2\in P_2$}
    \If {$\sem{ \gamma_1\land \gamma_2}\ne\emptyset$}
    \State $P'\gets P'\cup\{ \gamma_1\land \gamma_2\}$
    \State $f'( \gamma_1\land \gamma_2)\gets f_1( \gamma_1)\oplus f_2(
    \gamma_2)$
    \EndIf
    \EndFor
    \EndFor
    \State \Return $K$\textsc{Combine}($f'$)
    \EndFunction
    \Statex
    \Function{$K$Prod}{$f_1: P_1\to K, f_2: P_2\to K$}: $\GP K$
    \State \textbf{var} $f'$, $P'$
    \State $P'\gets \emptyset$
    \ForAll {$\gamma_1\in P_1$}
    \ForAll {$\gamma_2\in P_2$}
    \If {$\sem{ \gamma_1\land \gamma_2}\ne\emptyset$}
    \State $P'\gets P'\cup\{ \gamma_1\land \gamma_2\}$
    \State $f'( \gamma_1\land \gamma_2)\gets f_1( \gamma_1)\otimes
    f_2( \gamma_2)$
    \EndIf
    \EndFor
    \EndFor
    \State \Return $K$\textsc{Combine}($f'$)
    \EndFunction
    \Statex
    \Function{$K$Star}{$f: P\to K$}: $\GP K$
    \State \textbf{var} $f'$
    \ForAll {$\gamma\in P$}
    \State $f'( \gamma)\gets f( \gamma)^*$
    \EndFor
    \State \Return $K$\textsc{Combine}($f'$)
    \EndFunction
  \end{algorithmic}
  \caption{
    \label{fi:alg-ops}
    Functions which compute $\oplus$, $\otimes$ and $^*$ in $\GP K$.}
\end{figure}

Figure~\ref{fi:alg-ops} shows algorithms to compute these operations
in $\GP K$.  Note how these are similar to the \textsc{Split}
procedure in Fig.~\ref{fi:alg-minreach}.

Let $\nul, \one:\{ \ltrue\}\to K$ be the functions given
by $\nul( \ltrue)= 0$ and $\one( \ltrue)= 1$.  Then
$\nul, \one\in \GP K$.

\begin{lemma}
  \label{le:semop}
  Let $f_1, f_2\in \GP K$ and $p\in \px$. Then
  $\sem{ f_1\oplus f_2}( p)= \sem{ f_1}( p)\oplus \sem{ f_2}( p)$,
  $\sem{ f_1\otimes f_2}( p)= \sem{ f_1}( p)\otimes \sem{ f_2}( p)$,
  and $\sem{ f_1^*}( p)= \sem f_1( p)^*$.
\end{lemma}

\begin{lemma}
  \label{le:semeq}
  Let $f_1, f_2\in \GP K$.  Then $f_1= f_2$ iff $\sem{ f_1}= \sem{
    f_2}$.
\end{lemma}

\begin{proposition}
  \label{th:gpk-semiring}
  The structure $( \GP K, \oplus, \otimes, \nul, \one)$ forms a\linebreak
  $^*$-continuous Kleene algebra.
\end{proposition}

\begin{lemma}
  \label{le:sem_matrix}
  For $n\ge 1$, $M\in \GP K^{ n\times n}$, and $p\in \px$, $\sem{ M^*}(
  p)= \sem M( p)^*$.
\end{lemma}

We are ready to give the central result of this paper, stating that
for a given featured weighted automaton $\mcal F$, computing
$| \mcal F|$ suffices to obtain all projected values.

\begin{theorem}
  \label{th:comput-1}
  Let $\mcal F$ be a featured weighted automaton over $K$ and
  $p\in \px$.  Then $| \proj p{ \mcal F}|= \sem{| \mcal F|}( p)$.
\end{theorem}

\begin{proof}
  We have
  $\sem{| \mcal F|}( p)= \sem{ \alpha M^* \kappa}( p)= \sem \alpha( p)
  \sem M( p)^* \sem \kappa( p)$
  by Lemmas~\ref{le:semop} and~\ref{le:sem_matrix}.  Noting that the
  matrix representation of $\proj p{ \mcal F}$ is $( \sem \alpha( p),
  \sem M( p), k)$, the proof is finished.
\end{proof}


\section{Featured Energy Problems}
\label{se:energy}

In this final section we apply the theoretical results of this paper
to featured energy problems.

\subsection{Energy Problems}
\label{se:energy.prob}

The \emph{energy semiring}~\cite{DBLP:conf/atva/EsikFLQ13} is the
structure $\Nrg=( \E, \mathord\vee, \mathord\circ, \bot, \top)$.  Here
$\E$ is the set of \emph{energy functions}, which are partial
functions
$f: \Realnn\cup\{ \bot, \infty\}\to \Realnn\cup\{ \bot, \infty\}$ on
extended real numbers ($f( x)= \bot$ meaning that $f$ is undefined at
$x$) with the property that
\begin{equation}
  \label{eq:deriv1}
  \text{for all } x\le y: f( y)- f( x)\ge y- x\,.
\end{equation}
These have been introduced in~\cite{DBLP:conf/atva/EsikFLQ13} as a
general framework to handle formal energy problems as below.  The
operations in the semiring are (pointwise) maximum as $\oplus$ and
function composition as $\otimes$, and the neutral elements are the
functions $\bot$, $\id$ given by $\bot( x)= \bot$ and $\id( x)= x$ for
all $x\in \Realnn\cup\{ \bot, \infty\}$.

\begin{definition}
  An \emph{energy automaton} is a tuple $( S, I, F, T)$ consisting of
  a finite set $S$ of states, subsets $I, F\subseteq S$ of initial and
  accepting states, and a finite set $T\subseteq S\times \E\times S$
  of transitions.
\end{definition}

Hence the transition labels in energy automata are functions which
proscribe how a real-valued variable evolves along a transition.  An
\emph{energy problem} asks, then, whether some state is
reachable when given a certain \emph{initial energy}, or whether the
automaton admits infinite accepting runs from some initial energy:

A \emph{global state} of an energy automaton is a pair $q=( s, x)$
with $s\in S$ and $x\in \Realnn$.  A transition between global states
is of the form $(( s, x), f,( s', x'))$ such that $( s, f, s')\in T$
and $x'= f( x)$.  A (finite or infinite) \emph{run} of the automaton
is a (finite or infinite) path in the graph of global states and
transitions.

As the input to a decision problem must be in some way finitely
representable, we will state them for subclasses $\E'\subseteq \E$ of
\emph{computable} energy functions (but note that we give no technical
meaning to the term ``computable'' other that ``finitely
representable''); an $\E'$-automaton is an energy automaton
$( S, I, F, T)$ with $T\subseteq S\times \E'\times S$.

\begin{problem}[Reachability]
  \label{pb:reach}
  Given a subset $\E'\subseteq \E$ of computable functions, an
  $\E'$-automaton $\mcal S=( S, I, F, T)$ and a computable initial
  energy $x_0\in \Realnn$: do there exist $s_0\in I$ and a finite run
  of $\mcal S$ from $(s_0, x_0)$ which ends in a state in $F$?
\end{problem}

\begin{problem}[B{\"u}chi acceptance]
  \label{pb:buchi}
  Given a subset $\E'\subseteq \E$ of computable functions, an
  $\E'$-automaton $\mcal S=( S, I, F, T)$ and a computable initial
  energy $x_0\in \Realnn$: do there exist $s_0\in I$ and an infinite
  run of $\mcal S$ from $(s_0, x_0)$ which visits $F$ infinitely
  often?
\end{problem}

As customary, a run such as in the statements above is said to be
accepting.

\subsection{$^*$-Continuous Kleene $\omega$-Algebras}

We need a few algebraic notions connected to infinite runs in weighted
automata before we can continue.  An \emph{idempotent
  semiring-semimodule pair}~\cite{journals/sgf/EsikK07, book/BloomE93}
$( K, V)$ consists of an idempotent semiring
$K=( K, \oplus, \otimes, 0, 1)$ and a commutative idempotent monoid
$V=( V, \oplus, 0)$ which is equipped with a left $K$-action
$K \times V \to V$, $( x, v) \mapsto x v$, satisfying the following
axioms for all $x, y\in K$ and $u, v\in V$:
\begin{alignat*}{2}
  ( x\oplus y) v &= x v\oplus y v \qquad\qquad&
  x( u\oplus v) &= x u\oplus x v \\
  ( x y) v &= x( y v) &
  0\otimes x &= 0 \\
  x\otimes 0 &= 0 &
  1\otimes v &= v
\end{alignat*}
Also non-idempotent versions of these are in use, but we will only
need the idempotent one here.

A \emph{generalized $^*$-continuous Kleene
  algebra}~\cite{DBLP:conf/dlt/EsikFL15} is an idempotent
semiring-semimodule pair $( K, V)$ where $K$ is a $^*$-continuous
Kleene algebra such that for all $x, y\in K$ and for all $v\in V$,
\begin{equation*}
  x y^* v= \bigoplus_{ n\ge 0} x y^n v\,.
\end{equation*}

A \emph{$^*$-continuous Kleene
  $\omega$-algebra}~\cite{DBLP:conf/dlt/EsikFL15} consists of a generalized
$^*$-continuous Kleene algebra $( K, V)$ together with an
\emph{infinite product} operation $K^\omega\to V$ which maps every
infinite sequence $x_0, x_1,\dotsc$ in $K$ to an element
$\prod_n x_n$ of $V$.  The infinite product is subject to the
following conditions:
\begin{itemize}
\item For all $x_0, x_1,\dotsc\in K$,
  $\prod_n x_n= x_0 \prod_n x_{ n+ 1}$.
\item Let $x_0, x_1,\dotsc\in K$ and $0= n_0\le n_1\le\dotsm$ a
  sequence which increases without a bound. Let
  $y_k= x_{ n_k}\dotsm x_{ n_{ k+ 1}- 1}$ for all $k\ge 0$.  Then
  $\prod_n x_n= \prod_k y_k$.
\item For all $x_0, x_1,\dotsc, y, z\in K$, we have
  $\prod_n( x_n( y\oplus z))= \bigoplus_{ x_0', x_1',\dotsc\in\{ y,
    z\}\;} \prod_n x_n x_n'$.
\item For all $x, y_0, y_1,\dotsc\in K$, 
  $\prod_n x^* y_n= \bigoplus_{ k_0, k_1,\dotsc\ge 0\;} \prod_n x^{ k_n}
  y_n$.
\end{itemize}

For any idempotent semiring-semimodule pair $( K, V)$ and $n\ge 1$, we
can form the matrix semiring-semimodule pair $( K^{ n\times n}, V^n)$
whose elements are $n\times n$-matrices of elements of $K$ and
$n$-dimensional (column) vectors of elements of $V$, with the action
of $K^{ n\times n}$ on $V^n$ given by the usual matrix-vector product.

When $( K, V)$ is a $^*$-continuous Kleene $\omega$-algebra, then
$( K^{ n\times n}, V^n)$ is a generalized $^*$-continuous Kleene
algebra~\cite{DBLP:conf/dlt/EsikFL15}.
By~\cite[Lemma~17]{DBLP:conf/dlt/EsikFL15}, there is an $\omega$-operation
on $K^{ n\times n}$ defined by
\begin{equation*}
  M^\omega_i= \bigoplus_{1\le k_1,k_2,\dotsc\le n} M_{ i, k_1} M_{ k_1, k_2}\dotsm
\end{equation*}
for all $M\in K^{ n\times n}$ and $1\le i\le n$.  Also, if $n\ge 2$
and
$M = \left[ \begin{smallmatrix} a & b \\ c &
    d \end{smallmatrix}\right]$,
where $a$ and $d$ are square matrices of dimension less than $n$, then
\begin{equation*}
  M^\omega = 
  \begin{bmatrix}
    ( a\oplus b d^* c)^\omega\oplus( a\oplus b d^* c)^* b d^\omega \\
    ( d\oplus c a^* b)^\omega\oplus( d\oplus c a^* b)^* c a^\omega
  \end{bmatrix} .
\end{equation*} 

We also need another matrix-$\omega$-power below.  Let $n\ge 2$,
$k< n$ and $M\in K^{ n\times n}$, and write
$M = \left[ \begin{smallmatrix} a & b \\ c &
    d \end{smallmatrix}\right]$ as above, with $a\in K^{ k\times k}$
top left $k$-by-$k$ part of $M$.  We define
\begin{equation*}
  M^{ \omega_k}=
  \begin{bmatrix}
    ( a\oplus b d^* c)^\omega \\
    d^* c( a\oplus b d^* c)^\omega
  \end{bmatrix}.
\end{equation*}

Let $( K, V)$ be a $^*$-continuous Kleene $\omega$-algebra and
$\mcal S=( S, I, F, T)$ a $K$-weighted automaton.  An \emph{infinite
  path} in $\mcal S$ is an infinite alternating sequence $\pi=( s_0,
x_0, s_1, x_1, s_2,\dotsc)$ of transitions $( s_0, x_0, s_1),$ $( s_1, x_1,
s_2),\dotsc \in T$.  The \emph{weight} of $\pi$ is the infinite
product $w( \pi)= \prod_n x_n\in V$.

An infinite path $\pi=( s_0, x_0, s_1, x_1,\dotsc)$ in $\mcal S$ is
said to be \emph{B{\"u}chi accepting} if $s_0\in I$ and the set
$\{ n\in \Nat\mid s_n\in F\}$ is infinite.  The \emph{B{\"u}chi value}
$\| \mcal S\|$ of $\mcal S$ is defined to be the sum of the weights of
all its B{\"u}chi accepting infinite paths:
\begin{equation*}
  \| \mcal S\|= \bigoplus\{ w( \pi)\mid \pi\text{ B{\"u}chi accepting
    infinite path in $\mcal S$}\}
\end{equation*}
Let $( \alpha, M, k)$ be the matrix representation of $\mcal S$.  It
can be shown~\cite{DBLP:conf/dlt/EsikFL15} that
\begin{equation*}
  \| \mcal S\|= \alpha M^{ \omega_k}\,.
\end{equation*}

\subsection{Featured Energy Problems}

Recall that $\E$ denotes the set of energy functions: functions
$f: \Realnn\cup\{ \bot, \infty\}\to \Realnn\cup\{ \bot, \infty\}$ with
the property~\eqref{eq:deriv1} that whenever $x\le y$, then
$f( y)- f( x)\ge y- x$; and that
$\Nrg=( \E, \mathord\vee, \mathord\circ, \bot, \top)$ is the semiring
of energy functions.

\begin{lemma}[\cite{DBLP:journals/corr/EsikFL15a}]
  $\Nrg$ is a $^*$-continuous Kleene algebra.
\end{lemma}

Let $\Bool=\{ \lfalse, \ltrue\}$ be the Boolean lattice.  We say that
a function $f: \Realnn\cup\{ \bot, \infty\}\to \Bool$ is
\emph{$\infty$-continuous} if $f= \bot$ or for all
$X\subseteq \Realnn\cup\{ \bot, \infty\}$ with $\bigvee X= \infty$,
$\bigvee f( X)= \ltrue$.

Let $\V$ be the set of $\infty$-continuous functions
$f: \Realnn\cup\{ \bot, \infty\}\to \Bool$.  With operation $\vee$
defined by $( f\vee g)( x)= f( x)\lor g( x)$ and unit $\bot$ given by
$\bot( x)= \lfalse$ for all $x\in \Realnn\cup\{ \bot, \infty\}$,
$\Vrg=( \V, \vee, \bot)$ forms a commutative idempotent monoid.  Then
$( \Nrg, \Vrg)$ is an idempotent semiring-semimodule pair.

Define an infinite product $\E\to \V$ as follows: Let
$f_0, f_1,\dotsc \in \E$ be an infinite sequence and
$x\in \Realnn\cup\{ \bot, \infty\}$.  Let $x_0= f_0( x)$ and, for each
$k\ge 1$, $x_k= f_k( x_{ k- 1})$.  Thus $x_0, x_1,\dotsc$ is the
infinite sequence of values obtained by application of finite prefixes
of the function sequence $f_0, f_1,\dotsc$.  Then
$( \prod_n f_n)( x)= \lfalse$ if there is an index $k$ for which
$x_k= \bot$ and $( \prod_n f_n)( x)= \ltrue$ otherwise.

It can be shown~\cite{DBLP:journals/corr/EsikFL15a} that $\prod_n f_n$
is $\infty$-continuous for any infinite sequence
$f_0, f_1,\dotsc \in \E$, hence this defines indeed a mapping
$\E^\omega\to \V$.

\begin{lemma}[\cite{DBLP:journals/corr/EsikFL15a}]
  $( \Nrg, \Vrg)$ is a $^*$-continuous Kleene $\omega$-algebra.
\end{lemma}

Hence the energy problems stated at the end of
Sect.~\ref{se:energy.prob} can be solved by computing reachability
and B{\"u}chi values of energy automata:

\begin{proposition}[\cite{DBLP:journals/corr/EsikFL15a}]
  Let $\mcal S=( S, I, F, T)$ be an energy automaton and
  $x_0\in \Realnn$.
  \begin{itemize}
  \item There exist $s_0\in I$ and a finite run of $\mcal S$ from
    $(s_0, x_0)$ which ends in a state in $F$ iff
    $| \mcal S|( x_0)\ne \bot$.
  \item There exist $s_0\in I$ and an infinite run of $\mcal S$ from
    $(s_0, x_0)$ which visits $F$ infinitely often iff
    $\| \mcal S\|( x_0)= \ltrue$.
  \end{itemize}
\end{proposition}

We now define energy problems for featured automata.  Recall that $N$
denotes a set of features and $\px\subseteq 2^N$ a set of products
over $N$.

\begin{definition}
  A \emph{featured energy automaton} over $\px$ is a tuple
  $( S, I, F, T)$ consisting of a finite set $S$ of states, subsets
  $I, F\subseteq S$ of initial and accepting states, and a finite set
  $T\subseteq S\times \GP \E\times S$ of transitions.
\end{definition}

Hence transitions in featured energy automata are labeled with
(injective) functions from guard partitions to energy functions.

\begin{lemma}
  \label{le:fomega-E}
  For $f\in \E$, $f^\omega\in \V$ is given by
  \begin{equation*}
    f^\omega( x)=
    \begin{cases}
      \lfalse &\text{if } x= \bot\text{ or } f( x)< x\,, \\
      \ltrue &\text{otherwise}\,.
    \end{cases}
  \end{equation*}
\end{lemma}

\begin{definition}
  \label{de:omega_gpe}
  Let $f: P\to \E\in \GP \E$ and define $w': P\to \V$ by
  $w'( \gamma)= f( \gamma)^\omega$.  Let $w\in \GP \V$ be the
  canonicalization of $w'$, then we define $f^\omega= w$.
\end{definition}

\begin{lemma}
  \label{le:sem_matrix-o}
  For $n\ge 1$, $k< n$, $M\in \GP \E^{ n\times n}$, and $p\in \px$,
  $\sem{ M^\omega}( p)= \sem M( p)^\omega$ and
  $\sem{ M^{\omega_k}}( p)= \sem M( p)^{ \omega_k}$.
\end{lemma}

\begin{theorem}
  \label{th:family-energy}
  Let $\mcal F$ be a featured energy automaton and $p\in \px$.  Then
  $\| \proj p{ \mcal F}\|= \sem{\| \mcal F\|}( p)$.
\end{theorem}

\begin{proof}
  We have
  $\sem{\| \mcal F\|}( p)= \sem{ \alpha M^{ \omega_k}}( p)= \sem
  \alpha( p) \sem M( p)^{ \omega_k}$
  by Lemmas~\ref{le:semop} and~\ref{le:sem_matrix-o}.  As the matrix
  representation of $\proj p{ \mcal F}$ is
  $( \sem \alpha( p), \sem M( p), k)$, the result follows.
\end{proof}

\section{Conclusion}

We have introduced featured (semiring-) weighted automata and shown
that, essentially, verification of their properties can be reduced to
checking properties of weighted automata.  This is because, from a
mathematical point of view, a featured weighted automaton over a
semiring $K$ is the same as a weighted automaton over the semiring of
functions from products (sets of features) to $K$.

Representing functions from products to $K$ as injective functions
from partitions of the set of products to $K$, we have exposed
algorithms which will compute featured weighted reachability in case
$K$ is a $^*$-continuous Kleene algebra.  It is easy to see that these
extend to the non-idempotent case of $K$ being a \emph{Conway
  semiring}.  The essence in our approach does not lie in these
technical details, but in the fact that we pass from $K$ to a semiring
of functions into $K$; this typically preserves properties one is
interested in.

We have also seen that energy properties are preserved when passing
from the weighted to the featured weighted setting; generally, if
$( K, V)$ is a $^*$-continuous Kleene $\omega$-algebra, then the
semiring-semimodule pair of functions from products to $K$ and $V$,
respectively, will also be such.

We are interested in extending the setting of this paper to other
weighted structures beyond semirings, for example the valuation
monoids of~\cite{DBLP:journals/ijfcs/DrosteM11}.  This will enable
feature-based treatment of properties such as limit-average cost and
will be useful for an extension to the timed setting
of~\cite{DBLP:conf/splc/CordySHL12}.  From a practical point of view,
we have shown in~\cite{DBLP:conf/splc/OlaecheaFAL16} that efficient
algorithms are available for the limit-average setting.

\newpage

\section*{Appendix: Proofs}

\begin{proof}[Proof of Lemma~\ref{le:canon}]
  To see that $\tilde f$ is injective, let
  $\tilde \gamma_1, \tilde \gamma_2 \in \tilde P$ and assume
  $\tilde f( \tilde \gamma_1)= \tilde f( \tilde \gamma_2)$.  Let
  $\Gamma_1, \Gamma_2\in P'$ such that
  $\tilde \gamma_1= \bigvee \Gamma_1$ and
  $\tilde \gamma_2= \bigvee \Gamma_2$, then
  $f( \Gamma_1)= f( \Gamma_2)$ and hence $\Gamma_1= \Gamma_2$,
  \ie~$\tilde \gamma_1= \tilde \gamma_2$.

  For the second claim, let $\gamma\in P$, then $\gamma\in \Gamma$ for
  some $\Gamma\in P'$, hence
  $\sem \gamma\subseteq \sem{ \bigvee \Gamma}$.  To see uniqueness,
  let $\tilde \gamma_1, \tilde \gamma_2 \in \tilde P$ and assume
  $\sem \gamma\subseteq \sem{ \tilde \gamma_1}$ and
  $\sem \gamma\subseteq \sem{ \tilde \gamma_2}$.  As
  $\sem \gamma\ne \emptyset$, this implies that
  $\sem{ \tilde \gamma_1}\cap \sem{ \tilde \gamma_2}\ne \emptyset$,
  hence $\tilde \gamma_1= \tilde \gamma_2$. 
\end{proof}

\begin{proof}[Proof of Lemma~\ref{le:gp-land-unique}]
  Existence of $\gamma_1$ and $\gamma_2$ is obvious by definition of
  $P_1\wedge P_2$.  For uniqueness, assume that there is $\gamma_1'\in
  P_1$ with $\gamma_1'\ne \gamma_1$ and $\gamma= \gamma_1'\land
  \gamma_2$.  Then $\gamma= \gamma_1\land \gamma_1'\land \gamma$, but
  as $P_1$ is a partition, $\sem{ \gamma_1\land \gamma_1'}= \sem{
    \gamma_1}\cap \sem{ \gamma_1'}= \emptyset$, hence $\sem \gamma=
  \emptyset$, a contradiction. 
\end{proof}

\begin{proof}[Proof of Lemma~\ref{le:semop}]
  Let $f_1: P_1\to K$ and $f_2: P_2\to K$.  Let $\gamma_1\in P_1$,
  $\gamma_2\in P_2$ be the unique feature guards for which $p\models
  \gamma_1$ and $p\models \gamma_2$, then $\sem{ f_1}( p)= f_1(
  \gamma_1)$ and $\sem{ f_2}( p)= f_2( \gamma_2)$.

  We have $p\in \sem{ \gamma_1\land \gamma_2}$, hence
  $\sem{ \gamma_1\land \gamma_2}\ne \emptyset$, so that
  $\gamma_1\land \gamma_2\in P_1\wedge P_2$.  Using the notation of
  Def.~\ref{de:op_gpk},
  $s'( \gamma_1\land \gamma_2)= f_1( \gamma_1)\oplus f_2( \gamma_2)=
  \sem{ f_1}( p)\oplus \sem{ f_2}( p)$.
  Write $s: P\to K$ and let $\tilde \gamma\in P$ be such that
  $\sem{ \gamma_1\land \gamma_2}\subseteq \sem{ \tilde \gamma}$,
  \cf~Lemma~\ref{le:canon}.  Then $p\models \tilde \gamma$, hence
  $\sem{ f_1\oplus f_2}( p)=( f_1\oplus f_2)( \tilde \gamma)= s'(
  \gamma_1\land \gamma_2)$.
  The proofs for $\otimes$ and~$^*$ are similar. 
\end{proof}

\begin{proof}[Proof of Lemma~\ref{le:semeq}]
  It is clear that $f_1= f_2$ implies $\sem{ f_2}= \sem{ f_2}$.  For
  the other direction, write $f_1: P_1\to K$ and $f_2: P_2\to K$.  For
  each $p\in \px$, let
  $\gamma_p= \bigland_{ f\in p} f\land \bigland_{ f\notin p} \neg f\in
  \Bool( N)$ denote its \emph{characteristic feature guard}; note that
  $\sem{ \gamma_p}=\{ p\}$.

  Let $P\in \GPart$ be the guard partition
  $P=\{ \gamma_p\mid p\in \px\}$, and define functions
  $f_1', f_2': P\to K$ by $f_1'( \gamma_p)= \sem{ f_1}( p)$ and
  $f_2'( \gamma_p)= \sem{ f_2}( p)$.  By definition, $f_1$ is the
  canonicalization of $f_1'$ and $f_2$ the canonicalization of
  $f_2'$.  By construction, $\sem{ f_1}= \sem{ f_2}$ implies $f_1'=
  f_2'$, hence $f_1= f_2$. 
\end{proof}

\begin{proof}[Proof of Prop.~\ref{th:gpk-semiring}]
  We show that the set $K^\px$ of functions from $\px$ to $K$ forms a
  $^*$-continuous Kleene algebra; the theorem is then clear from
  Lemmas~\ref{le:semop} and~\ref{le:semeq}.  For functions $\phi_1,
  \phi_2: \px\to K$, define $\phi_1\oplus \phi_2$, $\phi_1\otimes
  \phi_2$ and $\phi_1^*$ by $( \phi_1\oplus \phi_2)( p)= \phi_1(
  p)\oplus \phi_2( p)$, $( \phi_1\otimes \phi_2)( p)= \phi_1(
  p)\otimes \phi_2( p)$, and $\phi_1^*( p)= \phi_1( p)^*$.  Let $0,
  1: \px\to K$ be the functions $0( p)= 0$, $1( p)= 1$.  Then $(
  K^\px, \oplus, \otimes, 0, 1)$ forms an idempotent semiring.

  We miss to show $^*$-continuity.  Let $p\in \px$ and
  $\phi_1, \phi_2, \phi_3: \px\to K$, then
  \begin{align*}
    \big( \phi_1 \phi_2^* \phi_3\big)( p) &= \phi_1( p) \phi_2^*( p)
    \phi_3( p) \\
    &= \phi_1( p) \phi_2( p)^* \phi_3( p) \\
    &= \phi_1( p)\big(
    \bigoplus_{ n\ge 0} \phi_2( p)^n \big) \phi_3( p) \\
    &= \bigoplus_{ n\ge
      0} \phi_1( p) \phi_2( p)^n \phi_3( p) \\
    &= \bigoplus_{ n\ge 0} \phi_1(
    p) \phi_2^n( p) \phi_3( p) \\
    &= \big( \bigoplus_{ n\ge 0} \phi_1
    \phi_2^n \phi_3\big)( p)
  \end{align*}
\end{proof}

\begin{proof}[Proof of Lemma~\ref{le:sem_matrix}]
  As the formula for computing $M^*$ involves only additions,
  multiplications and stars, this is clear by
  Lemma~\ref{le:semop}. 
\end{proof}

\begin{proof}[Proof of Lemma~\ref{le:fomega-E}]
  The claim is clear for $x= \bot$, so let $x\ne \bot$.  If
  $f( x)\ge x$, then also $f^n( x)\ge x$ for all $n\ge 0$, hence
  $f^\omega( x)= \ltrue$ by definition.

  If $f( x)< x$, then $f( x)\le x- M$, with $M= x- f( x)> 0$.
  By~\eqref{eq:deriv1}, $f^n( x)\le x- n M$ for all $n\ge 0$, hence
  there must be $k\ge 0$ for which $f^k( x)= \bot$, whence
  $f^\omega( x)= \lfalse$.
\end{proof}

\begin{proof}[Proof of Lemma~\ref{le:sem_matrix-o}]
  The formulas for $M^\omega$ and $M^{ \omega_k}$ involve only
  additions, multiplications, stars, and $\omega$s.  Invoking
  Lemmas~\ref{le:semop} and~\ref{le:sem_matrix}, we see that the proof
  will be finished once we show that for $f\in \GP \E$, $\sem{
    f^\omega}( p)= \sem f( p)^\omega$.

  Write $f: P\to \E$ and let $\gamma\in P$ be the unique feature guard
  for which $p\models \gamma$.  Then $\sem f( p)= f( \gamma)$.  Using
  the notation of Def.~\ref{de:omega_gpe},
  $w'( \gamma)= \sem f( p)^\omega$.  Write $w: P'\to \V$ and let
  $\tilde \gamma\in P'$ be such that
  $\sem \gamma\subseteq \sem{ \tilde \gamma}$,
  \cf~Lemma~\ref{le:canon}.  Then $p\models \tilde \gamma$, hence
  $\sem{ f^\omega}( p)= f^\omega( \tilde \gamma)= w'( \gamma)$.
\end{proof}

\end{document}